\documentclass{article}
\usepackage[letterpaper, right=1in, left=1in, top=1in, bottom=1in]{geometry}

\usepackage{graphicx} 
\usepackage{amsmath,amssymb,amsfonts,amsthm}
\usepackage{cite}
\usepackage{xcolor}
\usepackage{tikz}
\usepackage{multirow}
\usepackage[title]{appendix}

\newcommand{\R}{\mathbb{R}}
\newcommand{\N}{\mathbb{N}}
\newcommand{\bmtx}{\begin{bmatrix}}
\newcommand{\emtx}{\end{bmatrix}}
\newcommand{\bsmtx}{\left[\begin{smallmatrix}}
\newcommand{\esmtx}{\end{smallmatrix}\right]}
\newcommand{\LTV}{\mathrm{LTV}}
\newcommand{\SF}{\mathrm{SF}}
\DeclareRobustCommand{\svdots}{
  \vbox{%
    \baselineskip=0.2\normalbaselineskip
    \lineskiplimit=0pt
    \hbox{.}\hbox{.}\hbox{.}%
    \kern-0.2\baselineskip
  }%
}
\newtheorem{theorem}{Theorem}
\newtheorem{lemma}{Lemma}
\newtheorem{definition}{Definition}

\newtheorem{corollary}{Corollary}

\title{\LARGE \bf Safety Filter for Robust Disturbance Rejection via Online Optimization}

\author{Joyce Lai and Peter Seiler
	\thanks{J. Lai and P. Seiler are with the Department of Electrical Engineering and Computer Science at the University of Michigan, Ann Arbor, MI 48109, USA. \tt\small \{joycelai,pseiler\}@umich.edu
    }
}

\date{April 2025}

\begin{document}
\maketitle

\begin{abstract}
Disturbance rejection in high-precision control applications can be significantly improved upon via online convex optimization (OCO).  This includes classical techniques such as recursive least squares (RLS) and more recent, regret-based formulations.  However, these methods can cause instabilities in the presence of model uncertainty.  This paper introduces a safety filter for systems with OCO in the form of adaptive finite impulse response (FIR) filtering to ensure robust disturbance rejection.  The safety filter enforces a robust stability constraint on the FIR coefficients while minimally altering the OCO command in the $\infty$-norm cost.  Additionally, we show that the induced $\ell_\infty$-norm allows for easy online implementation of the safety filter by directly limiting the OCO command.  The constraint can be tuned to trade off robustness and performance.  We provide a simple example to demonstrate the safety filter.
\end{abstract}

\section{Introduction}
\label{sec:intro}

This paper presents a safety filter for robust disturbance rejection via online optimization.  Online convex optimization (OCO) is a broad set of methods that can be used for disturbance rejection.  This includes classical techniques such as recursive least squares (RLS) (Section 2.2 of \cite{astrom94} or Section 9.4 of \cite{hayes96}) and other variants \cite{tsao94,orzechowski08,jiang95}. It also includes more recent regret-based formulations \cite{anava14OCO,hazan11OML,zinkevich03OCP,goel2023best,agarwal19}.  This is especially relevant in high-precision control applications such as satellite pointing where moving physical components cause disturbances that are neither purely stochastic nor worst case \cite{thiele23,orzechowski08}.  In these applications, $H_2$ and $H_\infty$ can incorporate known disturbance characteristics through the use of disturbance filters.  However, the disturbance spectrum is often unknown at the time of design and, in these situations the $H_2$ and $H_\infty$ controllers will have conservative performance.  Instead, OCO is used to learn the disturbance characteristics and compute a control command to reject the disturbance.  However, the OCO is typically designed assuming perfect knowledge of the plant dynamics.  This can lead to instability when there are small amounts of model uncertainty resulting in unsafe operating conditions.

In the realm of safety critical control, a popular method of encoding safety constraints is by use of the control barrier function (CBF).  This is relevant in autonomous vehicle and robotic applications where safety is tied to obstacle avoidance.  These kinds of safety constraints can be accounted for by defining a safe region and constructing a corresponding CBF.  The CBF effectively defines the set of safe control inputs that keep the system from entering unsafe regions.  This can be implemented as a safety filter which minimally alters the baseline control input while imposing the CBF as a point wise in time constraint \cite{ames17,ames19}.  Additional works on robust CBFs account for model uncertainties \cite{buch22}.

Our work focuses on designing a safety filter which can be implemented online for uncertain systems with OCO.  We start with a motivating example where RLS is used for adaptive disturbance rejection.  In this example, uncertainty causes the system to go unstable.  This motivates the need for the safety filter design.  We then describe a more general framework for systems with OCO  which are subject to disturbance and uncertainty.  Specifically, we consider the class of OCO that takes form as an adaptive FIR filter with time-varying coefficients.  The safety filter has two competing objectives: robust stability and disturbance rejection performance. This combines robust control techniques and CBF methods for safety critical control.

Our main contributions are the following. First, we use a scaled small gain condition and induced $\ell_\infty$-norm bounding property (Theorem 1 and Lemma 2 from \cite{lai24}, respectively) to define a safe (i.e. stable) set of FIR coefficients.  The safe set is defined by a bound on the adaptive FIR filter that satisfies the scaled small gain (i.e. robust stability) condition.  Second, we formulate the safety filter as a constrained minimization problem which computes the signal that minimally alters the unconstrained FIR filter output and restricts the FIR coefficients to the set of stable gains point wise in time.  Note that we use robust control techniques to encode safety via FIR coefficient constraints rather than constructing a CBF.  However, we use the minimal perturbation method from CBF literature to design the safety filter.  Third, we provide an explicit solution to the constrained minimization problem which can easily be implemented online without explicitly computing the optimal FIR coefficients.  Lastly, we revisit the motivating example to demonstrate that the safety filter ensures both robust stability and disturbance rejection.

\textit{Notation:}  Let $\N_+$ and $\R^n$ denote the set of nonnegative integers and real $n\times 1$ vectors, respectively. Discrete-time signals are given by vector-valued sequences, $u:\N_+\to\R^n$, where $u_t \in \R^{n}$ is the value at time $t$. The $\ell_p$-norm of a signal $u$ is defined as: $\|u\|_p = \left( \sum_{t=0}^\infty \|u_t\|_p^p \right)^{1/p}$
where $\|u_t\|_p = \left( \sum_{i=1}^n |u_t(i)|^p \right)^{1/p}$ is the vector $p$-norm, and $u_t(i)$ is the $i^\text{th}$ entry of $u_t$.
Let $\ell_p^n$ denote the set of signals with finite $\ell_p$-norms, i.e. $\ell_p^n = \{u:\|u\|_p<\infty\}$.  The superscript $n$ is used to denote the dimension of the signal at any given time but may be dropped for simplicity.  Let the set $\ell_{pe}^n\subset\ell_p^n$ denote the subset of signals which have a finite $\ell_p$-norm on all finite time intervals, i.e. $\ell_{pe}^n = \{u:\sum_{t=0}^T \|u_t\|_p^p<\infty,\,\forall\,T\in\N_+\}$.  We refer to $\ell_p^n$ and $\ell_{pe}^n$ as the signal space and extended signal space, respectively.  Let $G:\ell_{pe}^n\to\ell_{pe}^m$ denote systems that map input signals $u\in\ell_{pe}^n$ to output signals $v\in\ell_{pe}^m$.  The induced $\ell_p$-norm of $G$ is defined as: $\|G\|_{p\to p} = \sup_{0\neq u\in\ell_p} \frac{\|v\|_p}{\|u\|_p}$.
We use $\|u\|$ and $\|G\|$ to denote signal and system induced norms when the specific $p$-norm is not important.  Additionally, we reserve capital letters for systems, matrices, and constants and lowercase letters for signals and vectors. Lastly, we use shorthand $u_{i:j}$ to denote a subsequence of a signal $u$ from time $i$ to $j$: $u_{i:j} = {\small \bsmtx u_i \\ \svdots \\ u_j \esmtx}$.

\section{Motivation}
\label{sec:motivation}

\subsection{Adaptive FIR Disturbance Rejection}

Consider the feedback diagram in Figure~\ref{fig:afdr-feedback} with an unknown output disturbance.  The system has a baseline controller in the inner-loop and an Adaptive FIR Disturbance Rejection (AFDR) controller in the outer-loop. The objective of the AFDR is to estimate the disturbance and inject a synthetic reference signal to cancel the effect of the disturbance.  However, we show in this section that the AFDR can cause an instability in the presence of model uncertainty.  This motivates the safety filter design introduced in Section~\ref{sec:prelim}.


\begin{figure*}[ht!]
\centering
\scalebox{0.9}{
\begin{picture}(440,110)(20,-45)
    \thicklines 

    \put(40,6.66){\vector(1,0){20}}
    \put(20,-6.66){\vector(1,0){40}}
    \put(60,-20){\framebox(50,40){$E(z)$}}
    \put(110,0){\vector(1,0){20}} 
        \put(116,6.66){$\hat{w}$}
    \put(130,-20){\framebox(50,40){RLS FIR}}
    \put(180,0){\vector(1,0){40}} 
        \put(206,6.66){$r$}
    \put(232.5,0){\circle{25}} 
    \put(245,0){\vector(1,0){20}} 
        \put(251,6.66){$e$}
    \put(265,-20){\framebox(50,40){$K(z)$}}
    \put(315,0){\vector(1,0){20}} 
        \put(321,6.66){$u$}
    \put(335,-20){\framebox(50,40){$G(z)$}}
    \put(385,0){\vector(1,0){20}} 
        \put(391,6.66){$v$}
    \put(417.5,0){\circle{25}} 
    \put(430,0){\vector(1,0){30}} 
        \put(437.5,6.66){$y$}

    \put(440,0){\line(0,-1){40}}
    \put(440,-40){\line(-1,0){420}}
    \put(20,-40){\line(0,1){33.34}}
    \put(232.5,-40){\vector(0,1){27.5}}

    \put(190,0){\line(0,1){35}}
    \put(190,35){\line(-1,0){150}}
    \put(40,35){\line(0,-1){28.34}}
    \put(30,-30){\dashbox(170,75)}
    \put(15,50){Adaptive FIR Disturbance Rejection (AFDR)}
    \put(417.5,35){\vector(0,-1){22.5}} 
        \put(415,41.66){$d$}

    \put(223,-1.3){\tiny $+$}
    \put(229.85,-9){\tiny $-$}
    \put(408,-1.3){\tiny $+$}
    \put(414.75,6.1){\tiny $+$}
\end{picture}}
\caption{Feedback system with a baseline controller combined with an RLS-based adaptive FIR disturbance rejection controller.}
\label{fig:afdr-feedback}
\end{figure*}
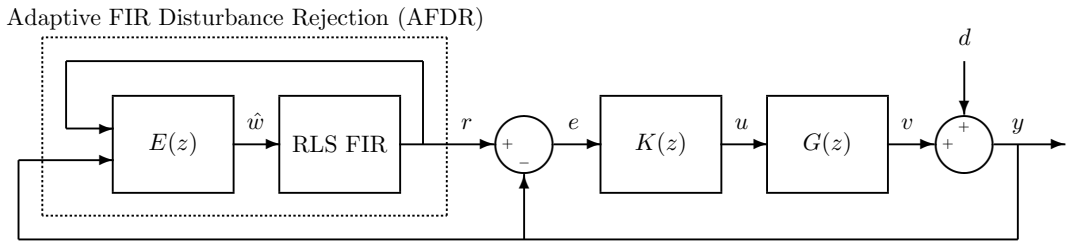

Let $G(z)$ and $K(z)$ denote the plant and inner-loop controller, respectively. Moreover, let $R(z)$, $D(z)$, $Y(z)$ denote the $z$ transforms of the signals 
$r$, $d$, $y$, respectively.  Neglecting the AFDR, the dynamics from inputs $(r,d)$ to output $y$ are given by:
\begin{align}
\label{eq:ExampleY}
    Y(z) &= T(z) \, R(z) + S(z) \, D(z),
\end{align}
where $L(z)=G(z)K(z)$, $S(z) = \frac{1}{1+L(z)}$ and $T(z) = \frac{L(z)}{1+L(z)}$ are the loop, sensitivity, and complementary sensitivity associated with the inner-loop feedback, respectively. 

The first component of AFDR is disturbance estimation.  The inner-loop controller $K(z)$ partially attenuates the disturbance.  Let $W(z)=S(z)\, D(z)$ denote the effective disturbance remaining at the output after the inner-loop is closed.  We can reconstruct the effective disturbance as $W(z) = Y(z) - T(z)R(z)$ using the $z$-domain relationship~\eqref{eq:ExampleY}.  This reconstruction is perfect when the true plant dynamics are perfectly known but will be imperfect otherwise.  In general, the true plant is not perfectly known and a plant model or estimate is used for the control design instead.  To make this distinction, let $G(z)$ and $\hat{G}(z)$ denote the uncertain and nominal plant, respectively.  Using the nominal plant model, we define the disturbance estimator as:
\begin{align}
\label{eq:estimator-z}
    \hat{W}(z) = E(z) \bmtx R(z) \\ Y(z) \emtx \mbox{ where }
    E(z) = \bmtx -\hat{T}(z) \\ I \emtx
\end{align}
where $\hat{W}(z)$ and $E(z)$ are the estimated effective disturbance and estimator, respectively.  Moreover, $\hat{T}(z)=\frac{\hat{G}(z)K(z)}{1+\hat{G}(z)K(z)}$ is an estimate of the complementary sensitivity constructed based on the nominal plant model $\hat{G}(z)$.

The second component of AFDR is adaptive FIR filtering.  Here, the effective disturbance estimate is used for further attenuation.  Let $\hat{w}_t\in\R$ denote the effective disturbance estimate at time $t$.  The injected reference $r_t$ is the output of an adaptive FIR filter with time-varying coefficients:
\begin{align}
\label{eq:AdaptiveFIR}
    r_t &= \sum_{i=0}^{H-1} \theta_{t}(i) \, \hat{w}_{t-i},
\end{align}
where $H$ is the adaptive FIR filter length and $\theta_{t}(i)\in\R$ is the FIR coefficient corresponding to $\hat{w}_{t-i}$ at time $t$.  The adaptive FIR filter~\eqref{eq:AdaptiveFIR} is similar to the FIR disturbance action policies used in recent OCO methods \cite{anava14OCO,hazan11OML,zinkevich03OCP,hazan07,shalev11,hazan16,agarwal19ANIPS,foster20,goel2023best}.

The goal of AFDR is to choose FIR coefficients $\theta_t := \bsmtx \theta_t(0) & \ldots & \theta_t(H-1) \esmtx^\top \in \R^H$ given the full history of disturbance estimates to minimize the variance of the output $y$.  This can be formulated as the following least squares optimization problem:
\begin{align}
\label{eq:LeastSquares}
    \theta_t^\star := \arg \min_{\theta\in \R^H}
    \left\| 
        \Phi_t(\hat{w}_{0:t}) \, \theta + \hat{w}_{0:t}
    \right\|_2,
\end{align}
where $\Phi_t(\hat{w}_{0:t}) \in \R^{(t+1) \times H}$ and $\hat{w}_{0:t}\in\R^{t+1}$ are the matrix of regressors and observation history at time $t$, respectively.  Appendix~\ref{app:AFDR} provides the details on the construction of $\Phi_t(\hat{w}_{0:t})$.  The least squares formulation~\eqref{eq:LeastSquares} determines the constant FIR coefficients that would have minimized the output variance given the past history of disturbance estimates.  This can be efficiently solved in real-time using RLS (Section 2.2 of \cite{astrom94} or Section 9.4 of \cite{hayes96}).  The adaptive FIR filter~\eqref{eq:AdaptiveFIR} then uses the RLS solution at each time: $\theta_{t} = \theta_t^\star$.

\subsection{Example: Effect of Model Uncertainty}
\label{sec:motivation-example}

To illustrate the effect of model uncertainty, consider the following nominal plant and controller:
{\small
\begin{align}
\nonumber
    \hat{G}(z) &= 10^{-4} \left(
        \frac{5.399 z^3 + 5.308 z^2 + 3.143 z + 4.459}
        {z^4 - 2.14 z^3 + 2.249 z^2 - 2.08 z + 0.9704}
        \right), \\
\label{eq:ExampleG0K}
    K(z) &=
        \frac{75.78 z^2 - 148.4 z + 72.63}
        {z^2 - 1.535 z + 0.5353}.
\end{align}
}

\noindent
The nominal plant corresponds to a continuous-time system with a double integrator and large resonance at 150 rad/sec.  This is a model of rigid body motion coupled with flexible dynamics as is common in many high precision feedback systems.  The controller corresponds to a PID controller with approximate derivative, designed to have a loop bandwidth near 12.5 rad/sec.  The continuous-time plant model and PID controller are discretized with sample time $T_s=0.01$ sec to obtain $\hat{G}(z)$ and $K(z)$.

The AFDR feedback system in Figure~\ref{fig:afdr-feedback} is simulated for 20 seconds (corresponding to $t=0,\ldots,\frac{20}{T_s}$ discrete time steps) with adaptive FIR filter length $H=8$.  The system is perturbed by the output disturbance:
\begin{align}
\label{eq:RLSdist}
    d_t = 1.4\sin(3t) + 0.9\sin(5t+0.4) + n_t,
\end{align}
where $n$ is IID, zero-mean white noise signal with variance $E[n^2_t]=(0.05)^2$.  Note that the white noise has a standard deviation of 0.05 which is a lower bound on the output standard deviation achievable via control.  Conversely, the disturbance has a standard deviation of 1.18.  This is what the output standard deviation would be with no inner- and outer-loop control, assuming the plant is stable.  Thus, we would like to reduce the standard deviation below 1.18 using both the inner- and outer-loop controllers.

The top subplot of Figure~\ref{fig:rls-motivation} shows the output for the nominal plant and controller provided above.  Note that this is a simulation for the case when there is no model uncertainty:  $G(z)=\hat{G}(z)$.  The AFDR is off for $t< 10$ seconds (corresponding to $r_t=0$ for $t=0,1,\ldots, \frac{10}{T_s}-1$).  The output $y$ has a standard deviation of 0.2734 during this time.  In other words, the classical controller is able to partially attenuate the disturbance.  The AFDR is on for $t\ge 10$, further reducing the output standard deviation down to 0.0647.  In other words, the AFDR almost perfectly cancels the two disturbance harmonics in~\eqref{eq:RLSdist}.

The bottom subplot of Figure~\ref{fig:rls-motivation} shows the output for the controller provided above and an uncertain plant given by:
\begin{align*}
    \Delta(z) &= 10^{-4} \left(
        \frac{0.5366 z - 1.195}
        {z^2 + 0.1429 z - 0.2798}
        \right), \\
    G(z) &= \hat{G}(z)+ \Delta(z),
\end{align*}
where $\Delta(z)$ represents the uncertainty or model error.  Here, the true plant dynamics used to evolve the states are $G(z)$, but the AFDR uses the nominal model $\hat{G}(z)$ to construct the estimated complementary sensitivity $\hat{T}(z)$ for the disturbance estimator in~\eqref{eq:estimator-z}.  The model error has minimal effect on the classical controller performance ($t< 10$).  However, the model error causes an instability once the AFDR is turned on ($t\ge 10$).  This illustrates the need for a framework for systems with online learning, uncertainty, and disturbance, as well as a method for robust AFDR.

\begin{figure}[ht!]
    \centering
    \includegraphics[width=0.5\linewidth]{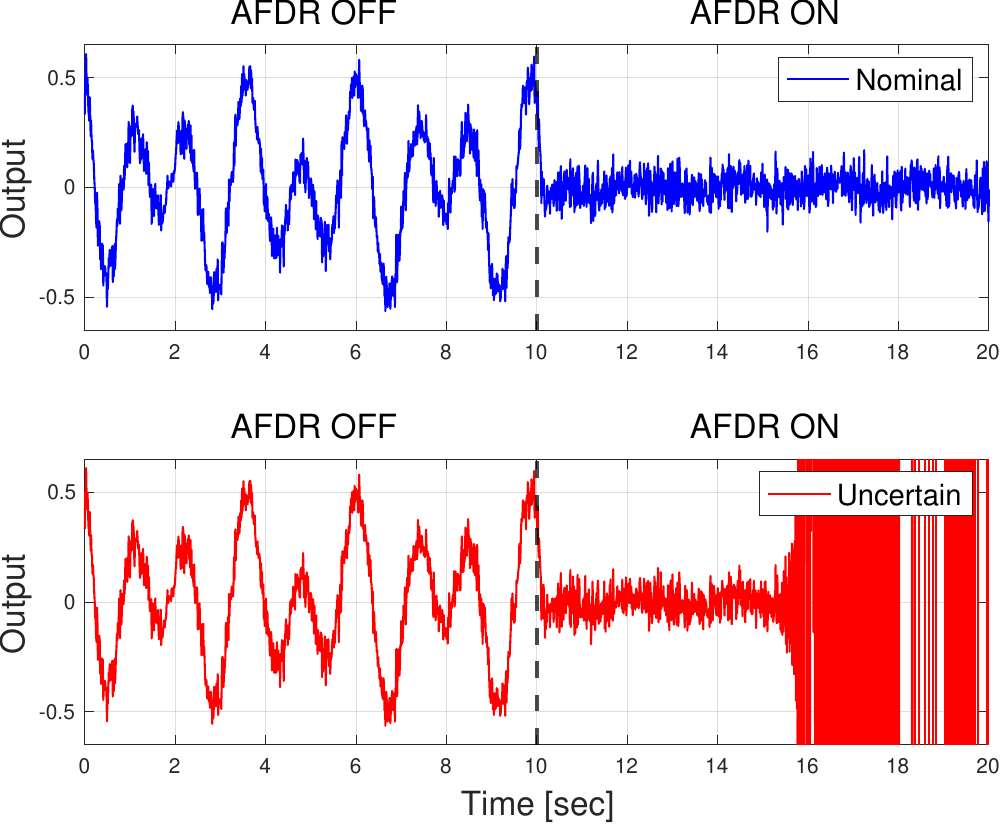}
    \caption{RLS-based AFDR rejects the disturbance at the output for the nominal plant (top), but goes unstable for the uncertain plant (bottom).}
    \label{fig:rls-motivation}
\end{figure}

\section{Preliminaries}
\label{sec:prelim}

\subsection{Problem Formulation}
\label{sec:problem}

The example in Section~\ref{sec:motivation} is based on a SISO model and updates the FIR coefficients via RLS.  We showed that small amounts of uncertainty can cause the system to go unstable.  This section focuses on the design of a safety filter for robust AFDR in a more general setting. This includes MIMO systems and alternative FIR update methods.

Consider the feedback system in Figure~\ref{fig:lft-uncertain} with disturbance $d$ and output $y$.  Let $M$ and $\Delta$ denote the nominal system dynamics and uncertainty, respectively.  We assume the uncertainty $\Delta$ is stable and bounded by $\delta$, i.e. $\|\Delta\|\le\delta$.  The filter $F=F_\mathrm{SF} F_\mathrm{LTV}$ describes the series interconnection of the adaptive FIR filter $F_\mathrm{LTV}$ into the safety filter $F_\mathrm{SF}$.  This feedback interconnection is called a linear fractional transformation (LFT) \cite{zhou95} and is denoted by $T_{d\to y}(M,\Delta,F)$ where $(\Delta,F)$ are wrapped around the upper channels of $M$.  We refer to $T_{d\to y}(M,\Delta,F)$ as the uncertain system and $T_{d\to y}(M,0,F)$ as the nominal system.  The dimensions of all signals are denoted by a subscript, e.g. $d_t$ and $y_t$ have dimensions $n_d\times 1$ and $n_y\times 1$, respectively.  Note that the LFT generalizes to alternative control architectures, but the signals are labeled corresponding to the AFDR feedback system in~\ref{fig:afdr-feedback} for comparison.

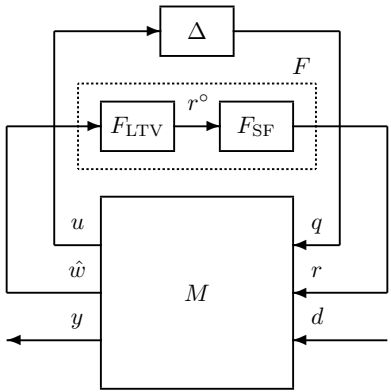
\begin{figure} 
\centering
\scalebox{0.9}{
\begin{picture}(160,160)(0,0)
    \thicklines 
    
    \put(40,20){\vector(-1,0){40}} 
        \put(27,26.66){$y$}
    \put(40,40){\line(-1,0){40}} 
        \put(26,46.66){$\hat{w}$}
    \put(40,60){\line(-1,0){20}} 
        \put(27,66.66){$u$}
    \put(40,0){\framebox(80,80){$M$}}
    \put(160,20){\vector(-1,0){40}} 
        \put(128,26.66){$d$}
    \put(160,40){\vector(-1,0){40}} 
        \put(128,46.66){$r$}
    \put(140,60){\vector(-1,0){20}} 
        \put(128,66.66){$q$}
    
    \put(0,110){\vector(1,0){40}} 
    \put(40,100){\framebox(30,20){$F_\mathrm{LTV}$}}
    \put(70,110){\vector(1,0){20}} 
        \put(76,116.66){$r^\circ$}
    \put(90,100){\framebox(30,20){$F_\mathrm{SF}$}}
    \put(120,110){\line(1,0){40}} 
    \put(30,92){\dashbox(100,36)} 
        \put(120,132){$F$}
    \put(20,150){\vector(1,0){45}} 
    \put(65,140){\framebox(30,20){$\Delta$}}
    \put(95,150){\line(1,0){45}} 
    \put(20,60){\line(0,1){90}} 
    \put(140,60){\line(0,1){90}} 
    \put(0,40){\line(0,1){70}} 
    \put(160,110){\line(0,-1){70}} 
\end{picture}
}
\caption{Uncertain system $T_{d\to y}(M,\Delta,F)$ with disturbance, adaptive FIR filtering, and safety filtering.}
\label{fig:lft-uncertain}
\end{figure}


As mentioned in Sections~\ref{sec:intro} and \ref{sec:motivation}, adaptive FIR filtering is useful for disturbance rejection for high-precision control applications where the system is affected by an unknown disturbance with learnable characteristics.  The adaptive FIR filter with filter length $H$ is defined as:
\begin{align}    
\label{eq:adaptive-fir}
        r^\circ_t &= \sum_{i=0}^{H-1}\theta_{t}(i) \, \hat{w}_{t-i},
\end{align}
where $\hat{w}_t\in\R^{n_w}$ and $r^\circ_t\in\R^{n_r}$ are the input and output of the adaptive FIR filter at time $t$, respectively.  We can express this compactly as $r^\circ_t =\Theta_t \, \hat{w}_{t:t-H+1}$ where $\Theta_t := \bsmtx \theta_{t}(0) & \ldots & \theta_{t}(H-1) \esmtx \in\R^{n_r\times n_wH}$.   The adaptive FIR filter has a systems interpretation which we denote as
$F_\mathrm{LTV}:\ell_p^{n_w}\to\ell_p^{n_r}$
where \eqref{eq:adaptive-fir} defines the output at time $t$.



The adaptive FIR coefficients $\Theta_t$ are typically updated via online optimization, e.g. online gradient descent (OGD) or RLS, and are based on the nominal dynamics $M$.  In order to prevent undesired consequences, e.g. profit loss or injury, high-precision applications require provable safety guarantees when the system dynamics are not perfectly known.  We use safety and stability interchangeably and use the following notion of stability.


\vspace{0.05in}
\begin{definition}[Nominal Finite-Gain Stability]
    The feedback interconnection in Figure~\ref{fig:lft-uncertain} is nominally finite-gain $\ell_p$ stable if $\|T_{d\to y}(M,0,F)\|<\infty$. 
\end{definition}
\vspace{0.05in}
\begin{definition}[Robust Finite-Gain Stability]
    The feedback interconnection in Figure~\ref{fig:lft-uncertain} is robustly finite-gain $\ell_p$ stable if $\max_{\|\Delta\|\le\delta} \|T_{d\to y}(M,\Delta,F)\|<\infty$.
\end{definition}
\vspace{0.075in}

Before stating the safety filter design objective, we make the following assumptions: (i) the disturbance is bounded, i.e. $d\in\ell_p$, (ii) the dynamics $M$ are known, LTI, and stable, and (iii) the uncertainty $\Delta$ is stable and bounded by $\delta$, i.e. $\|\Delta\|\le\delta$.  Given assumptions (i)-(iii) and uncertainty bound $\delta$, the objective is to design the safety filter to a) ensure robust finite-gain stability and b) preserve the nominal disturbance rejection performance.

\subsection{Background}
\label{sec:background}

In this section, we introduce a scaled small gain condition and induced $\ell_\infty$-norm bounding property of adaptive FIR filters.  These are existing results corresponding to Theorem 1 and Lemma 2 in \cite{lai24} for the case where there is no safety filter ($F_\SF=1$).  We use these results to formulate the safety filter in Section~\ref{sec:main}.

\vspace{0.05in}
\begin{theorem}[Scaled Small Gain]
\label{thm:scaled-SG}
    Let $T_{d\to y}(M,\Delta,F_\LTV)$ denote the feedback interconnection 
    in Figure~\ref{fig:lft-uncertain} with $F_\SF=1$.  Assume $M:\ell_{pe}\to\ell_{pe}$,
    $\Delta:\ell_{pe}^{n_u}\to\ell_{pe}^{n_q}$, and
    $F_\LTV:\ell_{pe}^{n_w} \to\ell_{pe}^{n_r}$
    are finite-gain stable systems of appropriate dimensions with $\|F_\LTV\|\le\beta$ and $\|\Delta\|\le\delta$.
    
    Let $M$ be partitioned as:
    \begin{align}
        M &= \bmtx M_{11} & M_{12} \\ M_{21} & M_{22} \emtx
    \end{align}
    where $M_{11}$ and $M_{22}$ have dimensions $(n_u+n_w)\times (n_q+n_r)$
    and $n_y\times n_d$, respectively. Moreover, define the following scaled system for any scalars $s_1$ and $s_2$:
    {\small
    \begin{align}
        \bar{M}_{11}(s_1,s_2,\delta,\beta) &:= \bmtx \frac{1}{s_1}I & 0 \\ 0 & \frac{1}{s_2}I \emtx M_{11} \bmtx s_1\delta I & 0 \\ 0 & s_2\beta I \emtx.
    \end{align}
    }
    
    \noindent The feedback interconnection $T_{d\to y}(M,\Delta,F_\LTV)$ is robustly finite-gain stable if there exists scalars $s_1>0$ and $s_2>0$ such that $\|\bar{M}_{11}(s_1,s_2,\delta,\beta)\|<1$.
\end{theorem}
\vspace{0.1in}

This result holds for any signal $p$-norm and corresponding system induced $\ell_p$-norm.  Assuming the  uncertainty bound $\delta$ is known, we can ensure robust finite-gain stability by finding a bound $\beta$ on the FIR filter $F_\LTV$ such that the scaled small gain condition holds.  
The bound $\beta$ roughly quantifies the amount of freedom the online optimization has to learn and reject the disturbance. As the bound increases, the OCO controller has potential for improved performance, but risks instability.  The largest such bound $\beta^\star$ can be computed by solving the optimization problem:
\begin{equation}
\begin{aligned}
\label{eq:beta-opt}
    \beta^\star =& \arg\sup_{\beta, s_1, s_2>0} & & \beta \\
    &\text{subject to} & & \|\bar{M}_{11}(s_1,s_2,\delta,\beta)\| < 1.
\end{aligned}
\end{equation}
Thus, any bound $\beta\in[0,\beta^\star)$ will ensure robust stability.  We use this result to define the safety filter.

The next result is useful for implementing the safety filter online.  Note again that Theorem~\ref{thm:scaled-SG} holds for any induced norm.  While the induced $\ell_2$-norm is a typical choice in robust control, it is easier to implement gain constraints on the adaptive FIR filter $F_\LTV$ online using the induced $\ell_\infty$-norm.  The following lemma relates the induced $\ell_\infty$-norm of the adaptive FIR filter to the induced $\infty$-norm of the adaptive FIR coefficients at each time.

\vspace{0.05in}
\begin{lemma}[Adaptive FIR Bounding Property]
\label{thm:FIR-bounding-prop}
    Suppose the adaptive FIR filter $F_\LTV$ has the output at each time $t$ given by~\eqref{eq:adaptive-fir}. Then
    \begin{align}
        \|F_\LTV\|_{\infty\to\infty} &= \sup_t\|\Theta_t\|_{\infty\to\infty},
    \end{align}
    where $\Theta_t := \bsmtx \theta_{t}(0) & \ldots & \theta_{t}(H-1) \esmtx \in\R^{n_r\times n_wH}$.
\end{lemma}
\vspace{0.1in}

\noindent Thus, we can bound the induced $\ell_\infty$-norm of the system $F_\LTV$ by imposing an induced $\infty$-norm constraint on the matrix $\Theta_t$ at each time $t$.  The next section gives the formal definition of the safety filter and its online implementation.

\section{Main Results}
\label{sec:main}

In this section, we define the safety filter as the solution to an online optimization problem using the results in Section~\ref{sec:background} and provide an explicit solution which can be easily implemented online.

The first objective of the safety filter is to ensure robust stability.  We will use the safety filter to impose this as a constraint on the FIR coefficients pointwise in time.  Let $\beta\in[0,\beta^\star)$ be the bound on $F=F_\SF F_\LTV$, i.e. $\|F\|_{\infty\to\infty}\le\beta$, where $\beta^\star$ is the solution to~\eqref{eq:beta-opt}.  We define the safe set of FIR coefficients as:
\begin{align}
\label{eq:safe-set}
    \mathcal{F}_\beta:=\{\Theta\in\R^{n_r\times n_wH}: \|\Theta\|_{\infty\to\infty}\le\beta\}.
\end{align}
If we design $F_\SF$ to enforce the constraint $\Theta_t\in\mathcal{F}_\beta$ for all $t$, then $\|F\|_{\infty\to\infty}\le\beta$ by Lemma~\ref{thm:FIR-bounding-prop}.  Moreover, the closed-loop system $T_{d\to y}(M,\Delta,F)$ is finite-gain stable by Theorem~\ref{thm:scaled-SG}.  There are many possible choices for the FIR coefficients that will satisfy the robust stability constraint $\Theta_t\in\mathcal{F}_\beta$ for all $t$. 

The second objective of the safety filter is to preserve the nominal disturbance rejection performance.  Let $(r_t^\circ,\Theta_t^\circ)$ and $(r_t,\Theta_t)$ denote the output and corresponding coefficients of the adaptive FIR filter $F_\LTV$ and safety filter $F_\SF$, respectively.  We refer to $r_t^\circ$ as the original or unconstrained OCO command and $r_t$ as the robust or constrained OCO command.  Assuming the coefficient update method is well designed, the original OCO command effectively rejects the disturbance without model uncertainty.  Thus, we are interested in designing the safety filter to enforce robust stability through the safe set $\mathcal{F}_\beta$ while minimizing the change in the original OCO command.  Considering both objectives, we define the safety filter as:
\begin{equation}
\label{eq:safety-filter}
\begin{aligned}
    (r^\star,\Theta^\star) =& \arg\min_{r,\Theta} & & \|r-r_t^\circ\|_{\infty} \\
    & \text{subject to} & & r = \Theta \, \hat{w}_{t:t-H+1} \\
    & & &\Theta\in\mathcal{F}_\beta.
\end{aligned}
\end{equation}
The safety filter output at time $t$ is then defined as $r_t=r^\star$.  Moreover, $\Theta_t=\Theta^\star$ corresponds to the FIR coefficients that are safe and generate the command $r_t=r^\star$.

The minimization problem~\eqref{eq:safety-filter} has an explicit solution in the $\infty$-norm cost.  We provide the solution and proof in the following theorem.

\vspace{0.05in}
\begin{theorem}[Safety Filter Solution]
\label{thm:SF-sol}
Let $i_0$ be an index such that $|\hat{w}_{t:t-H+1}(i_0)|=\|\hat{w}_{t:t-H+1}\|_\infty$ and $e_{i_0}$ be the $i_0^\mathrm{th}$ basis vector.  Then the explicit solution to~\eqref{eq:safety-filter} is:
\begin{align}
\label{eq:r-opt}
    r^\star &= \Theta^\star \, \hat{w}_{t:t-H+1},
\end{align}
where the $i^\mathrm{th}$ row of $\Theta^\star$ is defined as:
{\small
\begin{align}
\label{eq:Theta-opt}
    (\Theta^\star)_i &= \min\big(
    |r^\circ_t(i)|, \, \beta |\hat{w}_{t:t-H+1}(i_0)| \big)
    \cdot
    \frac{\mathrm{sign}(r^\circ_t(i))}{\hat{w}_{t:t-H+1}(i_0)}
    e_{i_0}^\top.
\end{align}
}
\end{theorem}

\begin{proof}
There are two cases to consider: (A) $\|r_t^\circ\|_\infty \le \beta \|\hat{w}_{t:t-H+1}\|_\infty$ and (B) $\|r_t^\circ\|_\infty > \beta \|\hat{w}_{t:t-H+1}\|_\infty$.   

First, consider Case (A). In this case, each 
entry of $r_t^\circ$ satisfies  $|r^\circ_t(i)|\le \beta |\hat{w}_{t:t-H+1}(i_0)|$ for $i=1,\ldots, n_r$.  Hence, Equation~\ref{eq:Theta-opt} simplifies to:
\begin{align}
\label{eq:Theta-optCaseA}
    (\Theta^\star)_i &= 
    \frac{r^\circ_t(i)}{\hat{w}_{t:t-H+1}(i_0)}
    \cdot
    e_{i_0}^\top.
\end{align}
Substitute this into \eqref{eq:r-opt} to obtain:
\begin{align}
r_i^\star =   \left(  \frac{r^\circ_t(i)}{\hat{w}_{t:t-H+1}(i_0)}
    \cdot
    e_{i_0}^\top \right) 
    \hat{w}_{t:t-H+1} = r^\circ_t(i).
\end{align}
Thus $r^\star = r_t^\circ$ yielding the cost $\|r^\star-r_t^\circ\|_\infty=0$. This is optimal since the cost must be nonnegative.  The safety filter leaves the FIR filter command unchanged in Case (A).

Next, consider Case (B).  We can  lower bound the optimal cost by noting that any
$(r,\Theta)$ feasible for
\eqref{eq:safety-filter} must satisfy:
\begin{align}
\label{eq:rinffUB}
\begin{split}
    \|r\|_\infty 
& \le \|\Theta\|_{\infty\to\infty} \cdot \|\hat{w}_{t:t-H+1}\|_\infty     \\
& \le \beta \cdot \|\hat{w}_{t:t-H+1}\|_\infty 
\end{split}
\end{align}
Equation~\ref{eq:rinffUB}, combined with triangle inequality, can be used to used to lower bound  the cost for
\eqref{eq:safety-filter}:
\begin{align}
\label{eq:CaseBLB}
\begin{split}
\|r-r_t^\circ\|_\infty  
& \ge \|r_t^\circ\|_\infty - \|r\|_\infty \\
& \ge \|r_t^\circ\|_\infty -\beta \|\hat{w}_{t:t-H+1}\|_\infty
\end{split}
\end{align}
We complete the proof by showing
that $(r^\star,\Theta^\star)$
in \eqref{eq:r-opt} and
\eqref{eq:Theta-opt}
achieve this lower bound.
Rewrite entry $i$ of~\eqref{eq:r-opt} as:
\begin{align}
\label{eq:r-opt-new}
    r^\star(i) &= c(i) \cdot \mathrm{sign}\left(r^\circ_t(i)\right)
\end{align}
where $c(i) = \min\left( |r^\circ_t(i)|, \beta |\hat{w}_{t:t-H+1}(i_0)| \right)$. We can express the cost for this $r^\star$ as:
\begin{align}
\label{eq:CaseBUB1}
    \|r^\star-r_t^\circ\|_\infty 
    &= \max_i \big| c(i)\cdot\mathrm{sign}\left( r^\circ_t(i) \right) - r^\circ_t(i) \big| \nonumber \\
    &= \max_i \big| |r^\circ_t(i)| - c(i) \big|
\end{align}
This can be simplified further based on the definition of $c(i)$:
\begin{align*}
    \|r^\star-r_t^\circ\|_\infty 
    &= \max_i \max\big\{ 0, |r_t^\circ(i)| - \beta |\hat{w}_{t:t-H+1}(i_0)| \big\}
\end{align*}
This implies that
$\|r^\star-r_t^\circ\|_\infty \le \|r_t^\circ\|_\infty -\beta \|\hat{w}_{t:t-H+1}\|_\infty$. In fact, this upper bound is achieved for at least one index $i$. Hence $(r^\star,\Theta^\star)$
in \eqref{eq:r-opt} and
\eqref{eq:Theta-opt}
achieve the lower bound~\eqref{eq:CaseBLB}
and are optimal.
\end{proof}

\vspace{0.1in}

Theorem~\ref{thm:SF-sol} provides the explicit expression of the safety filter output, i.e. robust OCO command, at each time.  This constrained OCO command imposes the scaled small gain condition (Theorem~\ref{thm:scaled-SG}) for stability and minimally alters the original OCO command.  Next, we state a simple corollary of Theorem~\ref{thm:SF-sol} that allows us to directly compute the safety filter output without explicitly computing the optimal FIR coefficients.

\vspace{0.05in}
\begin{corollary}
\label{thm:SF-update}
Let $r_t^\circ$ and $r_t$ denote the output of $F_\LTV$ and $F_\SF$ at time $t$, respectively.  The safety filter output has the following explicit expression that does not depend on the optimal adaptive FIR coefficients $\Theta^\star$.
\begin{align}
    r_t(i) =
    \begin{cases}
        r^\circ_t(i) & |r^\circ_t(i)| \le r_\mathrm{max} \\
        r_\mathrm{max}\cdot\mathrm{sign}(r^\circ_t(i)) & |r^\circ_t(i)| > r_\mathrm{max}
    \end{cases}
\end{align}
where $r_\mathrm{max} = \beta \|\hat{w}_{t:t-H+1}\|_\infty$ is the largest possible value of each element of $r_t$.
\end{corollary}
\vspace{0.1in}

At each time $t$, we can simply use Corollary~\ref{thm:SF-update} to compute the safety filter output or robust OCO command $r_t$ without explicitly computing the optimal coefficients $\Theta^\star$.  The next section illustrates the effect of the safety filter.

\section{Application to RLS}
\label{sec:application}

In this section, we revisit the motivating example in Section~\ref{sec:motivation} to illustrate the effect of the safety filter.  Here, we consider the same nominal plant $\hat{G}(z)$ and inner-loop controller $K(z)$ in \eqref{eq:ExampleG0K} with sample time $T_s=0.01$ seconds.
 
We assume an uncertainty bound of $\delta=3\times 10^{-4}$ and that the true plant $G(z)$ lies in the additive uncertainty set:
\begin{align}
    \mathcal{G}_\delta :=
        \{ G(z) = \hat{G}(z) + \Delta(z) : \|\Delta\|_{\infty\to\infty}\le\delta \}.
\end{align}
Note that $\delta=3\times 10^{-4}$ is consistent with the induced $\ell_\infty$-norm bound of the specific uncertainty used in the motivating example. We then construct 
the LFT $T_{d\to y}(M,\Delta,F)$ in Figure~\ref{fig:lft-uncertain} and solve the optimization problem~\eqref{eq:beta-opt} with $\delta=3\times 10^{-4}$.  This yields $\beta^\star=4.651$.  Next, we choose $\beta=2.8<\beta^\star$ to define the safe set $\mathcal{F}_\beta$ in~\eqref{eq:safe-set}.  This was tuned to roughly achieve the smoothest output.  The RLS-based adaptive FIR filter has filter length $H=8$, and the disturbance in~\eqref{eq:RLSdist} enters at the plant output.  Again, the disturbance has standard deviation 1.1776.  The AFDR system in Figure~\ref{fig:afdr-feedback} with the additional safety filter is simulated for 20 seconds ($t=0,\ldots,\frac{20}{T_s}$ discrete time steps), and the output is shown in Figure~\ref{fig:rls-robust}.

The top subplot of Figure~\ref{fig:rls-robust} shows the output of RLS-based AFDR with the safety filter for the nominal plant.  The AFDR is off for $t<10$ seconds, i.e. $r_t=0$ for $t=0,\ldots,\frac{10}{T_s}-1$.  The disturbance is partially attenuated by the classical controller resulting in an output standard deviation of 0.2734.  This is roughly the same as having no safety filter in Section~\ref{sec:motivation-example}.  The AFDR is on for $t\ge10$ seconds, and the disturbance is further attenuated.  During this time, the output has a standard deviation of 0.0876.  Note that the standard deviation is slightly higher than in the motivating example due to the conservativeness of the safety filter without uncertainty.  Regardless, AFDR with the safety filter still improves upon the classical controller to further cancel the disturbance.

The bottom subplot of Figure~\ref{fig:rls-robust} shows the overlapping outputs of RLS-based AFDR with the safety filter for 100 uncertain plants.  Here, the uncertain plants $\{G_i(z)\}_{i=1}^{100} \subset  \mathcal{G}_\delta $ correspond to 100 randomly generated uncertainties $\{\Delta_i(z)\}_{i=1}^{100}$ which satisfy $\|\Delta_i\|_{\infty\to\infty}\le\delta$.  Again, the AFDR is off for $t<10$ seconds, and the disturbance is partially attenuated by the classical controller.  Across the 100 uncertain plants, the output has an average standard deviation of 0.2734 (minimum of 0.2732, maximum of 0.2735) which aligns closely with the nominal performance with and without the safety filter.  The AFDR is turned on for $t\ge10$ seconds, and the disturbance is further attenuated without causing instability.  Across the 100 uncertain plants, the output has an average standard deviation of 0.0920 (minimum of 0.0798, maximum of 0.1347).  This aligns roughly with the nominal performance with and without the safety filter, illustrating that the safety filter has achieved both its objectives.

\begin{figure}[ht!]
    \centering
    \includegraphics[width=0.5\linewidth]{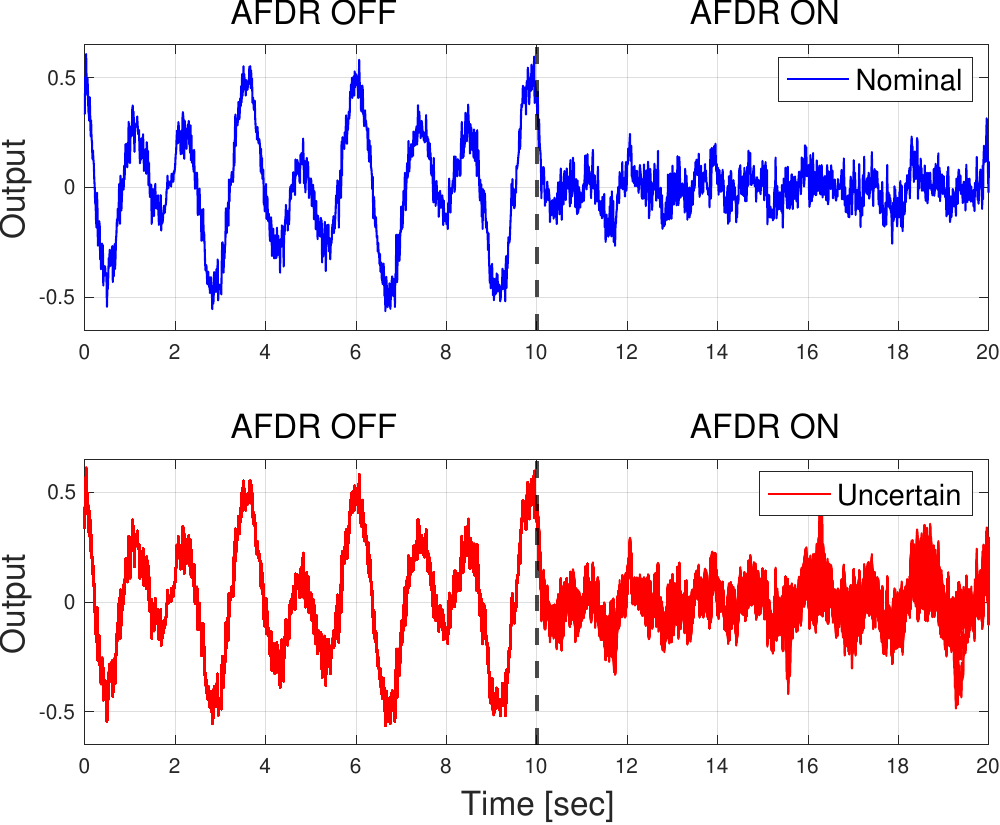}
    \caption{RLS-based AFDR with safety filter improves the disturbance rejection at the output for both the nominal plant (top) and 100 uncertain plants (bottom).}
    \label{fig:rls-robust}
\end{figure}

The effect of combining the classical controller, AFDR controller, and safety filter are summarized in Table~\ref{tab:output-sd}.  Again, the disturbance $d$ has standard deviation 1.1776 and white noise $n$ has standard deviation 0.05.  Column one corresponds to when the AFDR is off and only PID control is in effect ($t<10$). This is able to reject some, but not all of the disturbance. However, it shows good robustness to model uncertainty.  Column two corresponds to when the AFDR is turned on without safety filtering in addition to PID control ($t\ge10$).  This almost perfectly cancels the disturbance harmonics in the nominal case leaving only the effect of the white noise. However, it is sensitive to model error and can go unstable.  Column three corresponds to when the AFDR is turned on with safety filtering in addition to PID control ($t\ge10$).  The controller with the safety filter on the uncertain plants is relatively close in performance to the performance without the safety filter on the nominal plant.  Moreover, the safety filter ensures that the closed-loop remains stable even in the presences of model uncertainty.  Thus, the safety filter both maintains performance and ensures robustness.

\begin{table}[ht!]
\centering
\caption{Average Output Standard Deviation}
\label{tab:output-sd}
\begin{tabular}{|c||c|c|c|}
    \hline
        & PID & RLS-AFDR & Safety Filter \\
    \hline\hline
    Nominal (N=1)    & 0.2734 & 0.0647 & 0.0876 \\
    \hline
    Uncertain (N=100)  & 0.2734 & $+\infty$ (unstable) & 0.0920 \\
    \hline
\end{tabular}
\end{table}

\section{Conclusions}
In this paper, we present a safety filter for robust AFDR that enforces both robust stability and disturbance rejection performance.  The safety filter can be applied to systems where OCO control in the form of adaptive FIR filtering is used to improve the disturbance rejection.  We formulate the safety filter as an online optimization problem which restricts the FIR coefficients to a safe set while minimally altering the original OCO command in the $\infty$-norm cost.  We then provide an explicit solution to the optimization problem and show that the safety filter can be implemented by saturating the original OCO command without computing the optimal FIR coefficients.  Lastly, we provide a simple example to show that the safety filter ensures robustness and preserves disturbance rejection.  Future work will focus on integrating robustness and performance requirements into the online optimization used for the coefficient update as an alternative to safety filtering.

\section*{Acknowledgment}
This material is based upon work supported by the National Science Foundation under Grant No. 2347026. Any opinions, findings, and conclusions or recommendations expressed in this material are those of the authors and do not necessarily reflect the views of the National Science Foundation.

\bibliographystyle{IEEEtran}
\bibliography{Bibliography/RegretControl,Bibliography/RobControl,Bibliography/OCO,Bibliography/CBF}

\section*{Appendix}

\section*{A. AFDR Least Squares Formulation}
\label{app:AFDR}

This appendix provides details on the
least squares formulation given in Equation~\ref{eq:LeastSquares}.  
The output signal is given by $y=\hat{T} r + \hat{w}$ where $\hat{T}$ is the model of the complementary sensitivity and $\hat{w}$ is the estimate of the effective disturbance.  This can be rewritten as $y=m+\hat{w}$ with $m=\hat{T} r$.  Assume the complementary sensitivity $\hat{T}$ is modeled by the following state-space equation:
\begin{align}
\begin{split}    
    \hat{x}_{t+1} & = \hat{A} \, \hat{x}_t 
        +\hat{B} \, r_t, \,\,\, x_0=0 \\
    m_t & = \hat{C} \, \hat{x}_t 
        +\hat{D} \, r_t.
\end{split}
\end{align}
Then the signals $y$, $\hat{w}$, and $r$ can be stacked time 0 to time $t$. This gives the relation:
\begin{align}
\label{eq:LSy}
    y_{0:t} = M_1 r_{0:t} + \hat{w}_{0:t},
\end{align}
where
\begin{align}
    M_1 = \bsmtx \hat{D} & 0 & \dots & 0 \\
    \hat{C}\hat{B} & \hat{D} & \dots & 0 \\
    \vdots & \vdots & \ddots & \vdots \\
    \hat{C}\hat{A}^{t-1} \hat{B} & \hat{C}\hat{A}^{t-2} \hat{B} & \cdots 
         & \hat{D} \esmtx \in\R^{n_y(t+1)\times n_r(t+1)}.
\end{align}

The goal is to determine the reference inputs $r_{0:t}$ that would have minimized $\|y_{0:t}\|$ given the past data. More specifically, the injected reference signal is restricted to be the output of an FIR filter driven by $\hat{w}$:
\begin{align}
    r_t &= \sum_{i=0}^{H-1} \theta(i) \, \hat{w}_{t-i}.
\end{align}
Note that this is compactly expressed as $r_t=\Theta\,\hat{w}_{t:t-H+1}$ where $\Theta := \bsmtx \theta(0) & \cdots & \theta(H-1) \esmtx \in \R^{n_r\times n_wH}$.  The FIR coefficients $\Theta$ are assumed to be constant in this derivation.

The matrix $\Theta$ can be rearranged as a vector $\zeta:=\text{vec}(\Theta^\top)\in\R^{n_rn_wH}$ where $\text{vec}(\cdot)$ denotes columnwise stacking of the column vectors in $\Theta^\top$.  We can then express the FIR output as $r_t=(I_{n_r}\otimes \hat{w}_{t:t-H+1})^\top \zeta$.  The FIR relation can be stacked from times 0 to $t$ to obtain:
\begin{align}
\label{eq:LSr}
    r_{0:t} = M_2(\hat{w}_{0:t}) \, \zeta,
\end{align}
where
\begin{align}
    M_2(\hat{w}_{0:t}) = \bsmtx 
        (I_{n_r}\otimes \,\hat{w}_{0:-H+1})^\top \\
        \vdots \\
        (I_{n_r}\otimes \,\hat{w}_{t:t-H+1})^\top
    \esmtx \in\R^{n_r(t+1)\times n_rn_wH}.
\end{align}
Here we use the convention that $\hat{w}_j=0$ for $j< 0$.

Combine \eqref{eq:LSy} and \eqref{eq:LSr} to obtain the following expression for the stacked outputs:
\begin{align}
y_{0:t} = \big(M_1 M_2(\hat{w}_{0:t})\big) \, \zeta + \hat{w}_{0:t}.
\end{align}
Here $\Phi_t(\hat{w}_{0:t}):=M_1 M_2(\hat{w}_{0:t})$ contains the regressors that relate the FIR coefficients $\zeta$ to the past outputs $y_{0:t}$.
Thus the least squares problem at time $t$ is:
\begin{align}
    \min_{\zeta\in \R^{n_rn_wH}} \left\| 
     \Phi_t(\hat{w}_{0:t}) \, \zeta + \hat{w}_{0:t}
    \right\|_2.
\end{align}
This can be solved at each time $t$ via recursive least squares.  This gives the optimal FIR coefficients $\zeta^\star$ (or $\Theta^\star$ after rearranging) that would have minimized the output given the past data. The assumption is that the disturbance has some repeatable pattern such that $\Theta^\star$ will be a good choice for the FIR coefficients going into the future.

\section*{B. Offline Robust Stability Analysis}
\label{app:offline-analysis}

This appendix provides details for solving for the robust stability bound $\beta^\star$.  As mentioned in Section~\ref{sec:background}, the robust stability bound is solution to the optimization problem \eqref{eq:beta-opt} which has a system norm constraint.  In this paper, we are specifically interested in the system induced $\ell_\infty$-norm.  The optimization problem is convex for this case  ($p=\infty$) and can be formulated as a linear program (LP).

The general optimization problem is stated again here:
\begin{equation*}
\begin{aligned}
    \beta^\star =& \arg\sup_{\beta, s_1, s_2>0} & & \beta \\
    &\text{subject to} & & \|\bar{M}_{11}(s_1,s_2,\delta,\beta)\| < 1
\end{aligned}
\end{equation*}
where
\begin{align*}
    \bar{M}_{11}(s_1,s_2,\delta,\beta) &:= \bmtx \frac{1}{s_1}I & 0 \\ 0 & \frac{1}{s_2}I \emtx M_{11} \bmtx s_1\delta I & 0 \\ 0 & s_2\beta I \emtx
\end{align*}
is an $(n_u+n_w)\times(n_q+n_r)$ LTI system scaled by scalars $(s_1,s_2,\delta,\beta)$.  Note that $\delta\ge0$ is a pre-specified uncertainty level corresponding to $\|\Delta\|\le\delta$.

To simplify notation, first define the following system that includes the pre-specified uncertainty level $\delta$:
\begin{align}
    H:= M_{11} \bmtx \delta I & 0 \\ 0 & I \emtx,
\end{align}
where $H$ is also an $(n_u+n_w)\times (n_q+n_r)$ LTI system.  The optimization problem for $p=\infty$ can be rewritten as:
\begin{align}
\label{eq:beta-opt2}
\begin{split}    
    \beta^\star =& \arg\sup_{\beta, s_1, s_2>0} \,\, \beta \\
    &\text{subject to} \,\,\,\,
    \left\| \bmtx H_{11} & \left(\frac{s_2}{s_1}\right) \beta H_{12} \\
                  \left( \frac{s_1}{s_2}\right) H_{21} & \beta H_{22} \emtx
    \right\|_{\infty\to \infty} < 1,
\end{split}
\end{align}
where $H=\bsmtx H_{11} & H_{12} \\ H_{12} & H_{22} \esmtx$ is partitioned according to the block input and output dimensions.  For example, $H_{11}$ has dimensions $n_u\times n_q$.  Note that the variables $(s_1,s_2)$ only appear via the ratio $\frac{s_1}{s_2}$ and its inverse.  Thus, we can let $s_2=1$ without loss of generality.

Next, let $H_{11}(i,:)$ denote the system from all $n_q$ inputs to only the $i^{th}$ output of $H_{11}$ (where $i=1,\ldots,n_u$).  We will follow this notation to denote multiple input, single output (sub)systems.  It follows directly from the definition of the system induced $\ell_\infty$-norm that the inequality constraint in \eqref{eq:beta-opt2} can be equivalently written as follows:
\begin{align}
\label{eq:nu-constraints}
    & \left\| \bmtx H_{11}(i,:) &  \frac{\beta}{s_1} \, H_{12}(i,:) 
    \emtx \right\|_{\infty\to \infty}  < 1, \;\; \forall i=1,\ldots,n_u, \\
    & \left\| \bmtx s_1 \, H_{21}(j,:) &  \beta \, H_{22}(j,:) 
    \emtx \right\|_{\infty\to \infty}  < 1, \;\; \forall j=1,\ldots,n_w.
\end{align}
Since $s_1>0$, we can multiply both sides of \eqref{eq:nu-constraints} by $s_1$ and express the constraints as:
\begin{align*}
     \left\| \bmtx s_1 \, H_{11}(i,:) &  \beta \, H_{12}(i,:) 
    \emtx \right\|_{\infty\to \infty}  < s_1, \;\; \forall i= 1,\ldots,n_u.
\end{align*}
Furthermore, it follows again by definition of the system induced $\ell_\infty$-norm that we can express the lefthand side as the sum of norms.  Thus, we can rewrite the constraint as:
\begin{align*}
   & s_1 \big( \|H_{11}(i,:)\|_{\infty\to \infty} - 1 \big) + \beta \|H_{12}(i,:)\|_{\infty\to \infty} < 0, \;\;
   \forall i = 1,\ldots,n_u.
\end{align*}
Finally, the optimization problem for $p=\infty$ can be rewritten as following LP:
\begin{equation}
\label{eq:beta-LP}
\begin{aligned}
    \beta^\star =& \arg\max_{s_1,\beta>0} && \;\beta \\
    & \text{subject to} &&
    \bmtx
        \|H_{11}\|_{\infty\to\infty}-1 & \|H_{12}\|_{\infty\to\infty} \\
        \|H_{21}\|_{\infty\to\infty} & \|H_{22}\|_{\infty\to\infty}
    \emtx
    \bmtx s_1 \\ \beta \emtx
    <
    \bmtx 0 \\ 1 \emtx.
\end{aligned}
\end{equation}
This LP has 2 linear inequality constraints defined by the system induced $\ell_\infty$-norms of the partitions/subsystems of $H$.  The induced $\ell_\infty$-norm of a system can be computed by computing the $\ell_1$-norm of its impulse response.  Details are provided in Section 4.3 of \cite{dahleh95}.

\end{document}